\documentclass[pdflatex,sn-mathphys-num,Numbered]{sn-jnl}
\usepackage[varg]{txfonts}
\usepackage{epsfig,graphicx,natbib}
\usepackage{longtable}
\usepackage{amssymb}

\usepackage{amsmath}
\usepackage{mathrsfs} 
\usepackage{color}
\usepackage{subcaption}
\usepackage{booktabs}
\usepackage{xspace}
\usepackage{float}
\usepackage{longtable}
\usepackage{breqn}
\usepackage{hyperref}
\usepackage[raggedright]{sidecap}
\usepackage{comment}
\usepackage{todonotes}
\usepackage{booktabs,tabularx,threeparttable, multirow}
\usepackage{multicol}
 \usepackage[table]{xcolor}
\usepackage{booktabs,tabularx,threeparttable,multirow,array}
\usepackage{pifont}
\usepackage{colortbl}    
\usepackage{amssymb}     
\usepackage{booktabs}    
\usepackage{multirow}    
\usepackage{float}       
\usepackage{threeparttable} 
\usepackage{placeins}    
\usepackage{enumitem} 
\usepackage{algorithm}
\usepackage{algpseudocode}
\usepackage{makecell}
\usepackage{nameref}

\theoremstyle{thmstyleone}%
\newtheorem{theorem}{Theorem}
\newtheorem{proposition}[theorem]{Proposition}%
\newtheorem{remark}[theorem]{Remark}

\theoremstyle{thmstyletwo}%
\newtheorem{example}{Example}%

\theoremstyle{thmstylethree}%
\newtheorem{definition}{Definition}%
\newtheorem{lema}{Lemma}
\usepackage{tikz}
\usetikzlibrary{shapes.geometric,arrows.meta,positioning}

\newcommand{\Ln}{L_n}

\usepackage{etoolbox}

\newcommand{\setIntervals}[1]{\mathcal{I}(L_#1)}
\newcommand{\orderIntervals}{\preceq}
\newcommand{\orderIntervalsEstrict}{\prec}
\newcommand{\subsetdfn}{\mathcal{D}_1^{L_n\rightarrow Y_m}}
\newcommand{\setdfn}{\mathcal{D}_1^{L_n}}
\newcommand{\orderDfn}{\prec_{\Delta^\downarrow_\delta}}
\newcommand{\orderDfnEq}{\preceq_{\Delta^\downarrow_\delta}}
\newcommand{\setGdfn}{\mathrm{SDFN}}

\begin{document}

\title[]{An Efficient Computational Framework for Discrete Fuzzy Numbers Based on Total Orders}

\author[1]{\fnm{Arnau} \sur{Mir}}\email{arnau.mir@uib.es}

\author[1]{\fnm{Alejandro} \sur{Mus}}\email{alejandro.mus@uib.cat}

\author*[1]{\fnm{Juan Vicente} \sur{Riera}}\email{jvicente.riera@uib.es}

\equalcont{These authors contributed equally to this work.}

\affil*[1]{\orgdiv{Soft Computing, Image Processing and Aggregation Research Group (SCOPIA), Artificial Intelligence Research, Institute of the Balearic Islands (IAIB), Health Research Institute of the Balearic Islands (IdISBa)}, \orgname{Math and Comp. Sci. Department, University of Balearic Islands}, \orgaddress{\country{Spain}}}

\abstract{

Discrete fuzzy numbers, and in particular those defined over a finite chain 
\( L_n = \{0, \ldots, n\} \), have been effectively employed to represent 
linguistic information within the framework of fuzzy systems. Research on 
total (admissible) orderings of such types of fuzzy subsets, and specifically 
those belonging to the set \( \subsetdfn \) consisting of discrete fuzzy 
numbers \( A \) whose support is a closed subinterval of the finite chain 
\( L_n = \{0, 1, \ldots, n\} \) and whose membership values \( A(x) \), for 
\( x \in L_n \), belong to the set 
\( Y_m = \{ 0 = y_1 < y_2 < \cdots < y_{m-1} < y_m = 1 \} \), has facilitated the 
development of new methods for constructing logical connectives, based on a bijective  
function, called \textit{pos function},  that determines the position of each \( A \in \subsetdfn \).
For this reason, in this work we revisit the problem by introducing algorithms that exploit the combinatorial structure of total (admissible) orders to compute the \textit{pos} function and its inverse with exactness. The proposed approach achieves a complexity of $\mathcal{O}(n^{2}  m  \log n)$, which is quadratic in the size of the underlying chain ($n$) and linear in the number of membership levels ($m$). The key point is that the dominant factor is $m$, ensuring scalability with respect to the granularity of membership values.
The 
results demonstrate that this formulation substantially reduces computational 
cost and enables the efficient implementation of algebraic operations---such 
as aggregation and implication---on the set of discrete fuzzy numbers.

}


\maketitle

\section{Introduction}

The representation and manipulation of uncertainty remain central challenges in computational intelligence. Since Zadeh’s seminal formulation of fuzzy set theory~\citep{zadeh1975a, zadeh1975b, zadeh1975c}, multiple frameworks have been proposed to formalize linguistic imprecision and allow reasoning with words. Among them, the paradigm of \textit{Computing with Words} (CWW) has emerged as a powerful tool to model qualitative and imprecise information, finding applications in decision making, control systems, and knowledge representation \cite{zadehcww2012, GUPTA2022921, Zadeh1999}. In this regard, different methods based on multigranular fuzzy linguistics have been proposed in the literature \cite{MORENTEMOLINERA201549,Rodri_hesitant_2012}, establishing different categories depending on the fuzzy modeling used for linguistic expressions. Among these linguistic models, the one based on discrete fuzzy numbers \cite{voxman2001} stands out \cite{MASSANET2014,RIERA2015383} as well as their extensions, such as discrete or mixed fuzzy $Z$-numbers \cite{MASSANET202035,Aliev2015}. Among the key advantages of this linguistic model are the following \cite{morente2015, RIERA2015383}:
(i) it enables experts to express their preferences flexibly using different levels of granularity;
(ii) it eliminates the need to transform linguistic expressions prior to aggregation; and
(iii) it preserves information integrity throughout the aggregation process. 

However, the different methods proposed in the literature for aggregating or making inferences with this type of linguistic information are based on the use of aggregation or implication functions defined in a finite chain and on the use of the classical partial order interval  \cite{RIERA2013131,casasnovas2011,riera2012,riera2014}. The use of this type of partial order, for example in decision-making problems, implies that there may be expert opinions that are not comparable, which can pose a challenge when making a final decision. Therefore, the literature contains extensive research on the study of total orders and their construction methods, approached either from a theoretical perspective or applied to various computational linguistic models \cite{MaBeBu2022,SANTANA202044,BUSTINCE2013,Mir-Fuentes2024}. This type of total order has also been investigated in the context of discrete fuzzy numbers \cite{riera2021} and has been applied in the construction of implication and aggregation functions on the finite set \( \subsetdfn \) \cite{Gonzalez2025,DeMiguel2024}, which is the subset of discrete fuzzy numbers \( A \) whose support is a closed subinterval of the finite chain \( L_n = \{ 0, 1, \ldots, n \} \) and whose membership values \( A(x) \) for \( x \in L_n \) belong to the set \( Y_m = \{ y_1 = 0, y_2, \ldots, y_m = 1 \} \), with \( 0 = y_1 < y_2 < \cdots < y_{m-1} < y_m = 1 \).

The method for constructing these operators is based on a bijective function, called the \( pos \) \textit{function}, which is defined from the set \( \subsetdfn \) to the finite chain \( L_k \), where \( k = \left| \subsetdfn \right| - 1 \). This function provides the position occupied by each discrete fuzzy number \( A \in \subsetdfn \), according to a fixed total order, and was used in \cite{Gonzalez2025} to define implication functions on the set $\subsetdfn$, based on discrete implication functions defined on the finite chain \( L_k \). Furthermore, it was demonstrated that there is a one-to-one correspondence between the set of implication functions defined on $\subsetdfn$   and the set of discrete implications defined on $L_k$, verifying that the main properties (exchange principle, identity principle, etc.) were preserved. Following a similar reasoning, the  work \cite{DeMiguel2024} extends this framework to aggregation functions, providing a coherent and unified algebraic approach for operations defined on discrete fuzzy numbers.
From a theoretical perspective, these studies are noteworthy because they propose a method for constructing logical connectives based on total orders rather than partial orders, as discussed in \cite{riera2012,riera2014}.

Despite their theoretical relevance, the direct computation of the position function (\(\mathrm{pos}\)) or its inverse (\(\mathrm{pos}^{-1}\)) is computationally expensive, since the cardinality of the set
\[
\left| \subsetdfn \right| = \binom{n + 2m - 2}{2m - 2}
\]
depends both on the number of elements in the finite chain \(L_n\) (corresponding to the number of linguistic labels \cite{MASSANET2014}) and on the number of membership values "$m$'' used to represent the discrete fuzzy numbers \cite{Gonzalez2025,Mir-Fuentes2024}. In the literature \cite{Bonissone1986,Herrera2008}, linguistic chains with an odd number of labels are typically used, with the remaining labels arranged symmetrically around the central one, and the granularity usually limited to 11 or, at most, 13.

Classical approaches require enumerating or sorting all dfns according to the total (admissible) order, resulting in exponential time and memory complexity with respect to the number of membership levels. This computational bottleneck becomes critical when these mappings are invoked repeatedly, as in the computation of aggregation or implication functions in practical fuzzy reasoning systems.

The present work tackles this problem by proposing \textit{efficient algorithms for computing the} $\mathrm{pos}$ \textit{and} $\mathrm{pos}^{-1}$ \textit{functions} without generating the entire space of discrete fuzzy numbers. The proposed methods exploit the combinatorial structure of discrete fuzzy numbers and the recursive properties of admissible orders to achieve polynomial time complexity while maintaining exactness. In particular, both the ranking and unranking processes —namely, the computation of the function $\mathrm{pos}$ and its inverse— are performed in $\mathcal{O}(n^2 m \log n)$ time, requiring only constant additional memory. This represents a significant improvement over brute-force enumeration methods.

The remainder of this paper is structured as follows. Section~\ref{sec:preliminaries} reviews the theoretical background of discrete fuzzy numbers and orders. Section~\ref{sec:motivation} analyzes the computational challenges related to the $\mathrm{pos}$ and $\mathrm{pos}^{-1}$ functions. Section~\ref{sec:cost_inv_pos} presents the proposed efficient algorithms and their complexity analysis. Section~\ref{sec:empirical} provides an empirical validation of the theoretical results derived in the preceding section. Finally, Section~\ref{sec:conclusions} summarizes the main contributions and outlines directions for future research.

\section{Preliminaries}\label{sec:preliminaries}

This section reviews the most relevant concepts and results that will be used throughout the work.
\subsection{Partial and total orders}
\label{sub:partial_total_orders}
First, let us recall some definitions regarding partial and admissible orders, as well as some illustrative examples of them. 

\begin{definition}\cite{Gra2011}
Let $A$ a set and $\mathcal{R}\subset A\times A$ a binary relation defined over $A$. $\mathcal{R}$ is an order if it fulfills the following properties:
\begin{itemize}
    \item $\mathcal{R}$ is reflexive, that is $x\mathcal{R}x$, $\forall x\in A.$
    \item $\mathcal{R}$ is anti-symmetric, that is, given $x,y\in A$, if $x\mathcal{R} y$ and $y\mathcal{R} x$, then $x=y.$
    \item $\mathcal{R}$ is transitive, that is, given $x,y,z\in A$, if $x\mathcal{R}y$ and $y\mathcal{R}z$, then $x\mathcal{R} z.$

If for all $x, y \in A$ it holds that $x \mathcal{R} y$ or $y \mathcal{R} x$, then the relation $\mathcal{R}$ is called a \textit{total (lineal)} order; otherwise, it is called a \textit{partial order}.
\end{itemize}
\end{definition}
A partially ordered set(poset for short) is an ordered pair $(X,\leq)$ consisting of a set $X$ and a partial order on $X$. In particular, when the order is total, $(X,\leq)$ is called a totally ordered set.

 Let \(\Ln=\{0,1,\dots,n\}\) be a finite chain endowed with the natural order. 
Let us denote by $\setIntervals{n}$ the set of closed intervals defined on $L_n$:
$$
\setIntervals{n} = \{[a,b],\mid a,b\in L_n\}.
$$

\begin{definition}\label{deforder}\cite{BUSTINCE2013}
The order $\preceq$ over $\setIntervals{n}$ is called an \textit{admissible order}, if it satisfies
\begin{description}
\item[(i)] $\preceq$ is a total (linear) order on $\setIntervals{n}$,
\item[(ii)] For any two intervals \([a, b], [c, d] \in \setIntervals{n}\), the relation \([a, b] \preceq [c, d]\) holds whenever \([a, b] \leq_2 [c, d]\), where \(\leq_2\) represents the standard partial order on intervals, defined by \([a, b] \leq_2 [c, d]\) if and only if \( a \leq c \) and \( b \leq d \).
\end{description}
\end{definition}
The following binary relations are classical examples of total orders \cite{BUSTINCE2013,Mir-Fuentes2024}:
\begin{description}
    \item[Lexicographic Order 1:] $[a, b] \leq_{lex1} [c, d]$ if and only if $a < c$ or ($a = c$ and $b \leq d$).
    \item[Lexicographic Order 2:] $[a, b] \leq_{lex2} [c, d]$ if and only if $b < d$ or ($b = d$ and $a \leq c$).
    \item[Xu and Yager order:] $[a, b] \leq_{XY} [c, d]$ if and only if $(a + b) < (c + d)$ or (($a + b) = (c + d)$ and $(b - a) \leq (d - c)$).
    \item[t-inc order:]
$[a,b] \leq_{\mathrm{t-inc}} [c,d] \iff a <c  \lor (a = c \land d \leq b).$
\end{description}
We would like to point out that the above definition is a readaptation of the original definition established by Bustince et al \cite{BUSTINCE2013} which considered closed intervals of the unit interval instead of closed intervals defined on a finite chain $L_n$. This idea has already been used in literature, as can be seen in \cite{riera2021,Mir-Fuentes2024,Gonzalez2025}.
\subsection{Discrete fuzzy numbers}
Next, we outline the concept of a discrete fuzzy number and introduce the relevant notation.

\begin{definition}[\textbf{\cite{voxman2001}}]\label{dfn}
A fuzzy subset $A$ of $ \mathbb{R}$ with membership mapping $A:
\mathbb{R} \to [0,1]$ is called a \emph{discrete fuzzy number}, or \emph{dfn} for short, if
its support is finite, i.e., there exist $x_{1},...,x_{n}\in
\mathbb{R}$ with $x_{1}<x_{2}<...<x_{n}$ such that
$supp(A)=\{x_1,...,x_{n}\}$, and there are natural numbers $s, t$
with $1\leq s\leq t\leq n $ such that:
\begin{enumerate}
\item A($x_{i}$)=1 for all $i$ with $s\leq i\leq
t$. (\emph{core}) \item $A(x_{i})\leq A(x_{j})$ for all $i, j$ with  $1\leq i\leq j\leq s $. \item $A(x_{i})\geq
A(x_{j})$ for all $i, j$ with $t\leq i\leq j\leq n
$.
\end{enumerate}
A dfn \(A\)  with support $supp(A)=\{x_1,\ldots,x_n\}$ is denoted as \( A = \{ A(1)/x_1, \ldots, A(n)/x_n \} \).

\end{definition}

Recall that, for every \(\alpha \in (0,1]\), the \(\alpha\)-cut of the discrete fuzzy number \(A\) is given by
\[
A^\alpha = \left\{\, x_i \in \operatorname{supp}(A) \mid A(x_i) \geq \alpha \,\right\}.
\]

We will denote by \(\setdfn\) the set of all discrete fuzzy numbers (dfns) that have a closed subinterval of the finite chain \( L_n \) as their support. The importance of studying \( \setdfn \) lies in the capability of this family of discrete fuzzy numbers to act as linguistic expressions that accurately reflect an expert's opinions in a decision-making context \cite{MASSANET2014,RIERA2015383}.

Total and admissible orders on the set $\setdfn$ were investigated in \cite{riera2021}. Let us recall the essential ideas and concepts of its construction.

\begin{definition}[\textbf{\cite{riera2021}}]
Let \( A \in \setdfn \). Then, \( \alpha \in (0,1] \) is termed a relevant \( \alpha \)-level for \( A \) if there exists \( i \in \sup(A) \) such that \( A(i) = \alpha \).
\end{definition}

\begin{theorem}[\textbf{\cite{riera2021}}]
\label{ordrefuzzynumbers}

Let \( A, B \in \setdfn \) be two discrete fuzzy numbers, with their respective sets of relevant \(\alpha\)-levels given by
\[
S_A = \{ \alpha_1 < \cdots < \alpha_k = 1 \}, \quad \text{where } k \leq n + 1,
\]
\[
S_B = \{ \beta_1 < \cdots < \beta_m = 1 \}, \quad \text{where } m \leq n + 1.
\]
Let \( S_{AB} = S_A \cup S_B = \{ \gamma_1 < \gamma_2 < \cdots < \gamma_t = 1 \} \) be the union of both sets, where \( 1 \leq t \leq k + m - 1 \).

Let \( \orderIntervals \) be an admissible order on \( \setIntervals{n} \).

We define the binary relation:
\begin{itemize}
    \item \( A = B \) if and only if \( A^{\gamma_i} = B^{\gamma_i} \) for all \( i \in I = \{1, \ldots, t\} \).
    
    \item \( A \orderDfn B \) if \( A \neq B \) and there exists \( j \in I \) such that
    \[
    A^{\gamma_j} \orderIntervals B^{\gamma_j} \quad \text{and} \quad A^{\gamma_i} = B^{\gamma_i} \text{ for all } i > j.
    \]
    
    \item \( A \orderDfnEq B \) if \( A = B \) or \( A \orderDfn B \).
\end{itemize}

This binary relation \( \orderDfnEq \) is a total (admissible) order on \( \setdfn \).
\end{theorem}

\begin{remark}
Even though the total order mentioned in Definition \ref{ordrefuzzynumbers} depends on an admissible order on \( \setIntervals{n} \), the binary relation \( \orderDfnEq \) continues to be a total order when any total (not necessarily admissible) order on \( \setIntervals{n} \) is considered.
\end{remark}

We will denote by \( \subsetdfn \) the subset of discrete fuzzy numbers \( A \) whose support is a closed subinterval of the finite chain \( L_n \) and whose membership values \( A(x) \) for \( x \in L_n \) belong to the set \( Y_m = \{ y_1 = 0, y_2, \ldots, y_m = 1 \} \), with \( 0 = y_1 < y_2 < \cdots < y_{m-1} < y_m = 1 \).

Setting the number of possible values that the membership function can take, \( Y_m = \{ y_1 = 0, y_2, \ldots, y_m = 1 \} \), immediately implies that the cardinality of the set $\subsetdfn$ is finite. Therefore, in \cite{Mir-Fuentes2024}, this value is determined as shown in the following result.
\begin{proposition}\label{numerodedfn}

The number of discrete fuzzy numbers in the set \( \subsetdfn \) is given by
\[
\left| \subsetdfn \right| = \binom{n + 2m - 2}{2m - 2} = \frac{(n + 2m - 2)!}{(2m - 2)! \, n!}.
\]
\end{proposition}
In this way, if we consider a total (admissible) order constructed according to Theorem \ref{ordrefuzzynumbers}, \( \orderDfnEq \), we immediately obtain a bounded finite lattice $(\subsetdfn, \orderDfnEq, \mathbf{1}_{\min}, \mathbf{1}_{\max})$, which can be represented as (see \cite{Gonzalez2025})
\[
\mathbf{1}_{\min} \orderDfn A_2 \orderDfnEq \ldots \orderDfnEq \mathbf{1}_{\max}.
\]
This fact allows us to define the \textit{pos function} as follows:
\begin{definition}\cite{Gonzalez2025} \label{funcpos}Consider the finite chain $L_k=\{0,\ldots, k\}$ where $k=\left|\mathcal{A}_1^{L_n \times Y_m}\right|-1$. The function
$$
\begin{array}{rccl}
pos:&\subsetdfn &\longrightarrow&L_k\\
&A&\mapsto&pos(A)=\left|\{X\in \subsetdfn \mid\;X \orderDfnEq A\}\right|-1,
\end{array}
$$ 
where $\left|\{X\in \subsetdfn \mid\;X \orderDfnEq A\}\right|$ represents the number of elements in $\mathcal{A}_1^{L_n \times Y_m}$ that are less than or equal to 
$A$ (according to the total (admissible) order $ \orderDfnEq $), it is called the \textit{position function}. This function is called \textit{ranking dfn} because, given a dfn, it assigns it a rank within the finite chain $L_k$.
\end{definition}

The above function is an order isomorphism (see Proposition 3 in \cite{Gonzalez2025}) where 
its inverse function is :
\[
\begin{array}{rccl}
\mathit{pos}^{-1}: & L_k = \{0, \ldots, k\} & \longrightarrow & \subsetdfn \\
                   & i                       & \mapsto         & A \text{ such that } \mathit{pos}(A) = i,
\end{array}
\]
Consequently, given a discrete fuzzy number \( A\in \subsetdfn \) there exists a unique corresponding element in the finite chain $L_k$ and vice versa. In the same way, this function is called \textit{unranking dfn} because, given an index in $L_k$, it returns the dfn whose rank corresponds to that index.

This correspondence will be fundamental to the development of our research, as will be demonstrated in the following sections.

The following result proves that any discrete fuzzy number is completely determined from its membership values or, equivalently, from its relevant $\alpha$-cuts. The idea behind this representation using relevant alpha cuts is similar to Theorem 2.5 in \cite{KlirYuan1995}. We would like to highlight that next demonstration provides an algorithm that allows a discrete fuzzy number to be constructed from its relevant $\alpha$-cuts. This algorithm will be used in next sections.
\begin{lema} \label{descomposicio_dfn}
\label{EQALPHACUTS}
Let $A$ be a discrete fuzzy number in the set $\subsetdfn$, and let $S_A = \{y_{i_1} < \cdots < y_{i_k} = 1\}$ denote the set of relevant $\alpha$-levels of $A$.

The following statements are equivalent:
\begin{itemize}
    \item[(i)] The values $A(i)$ are known for all $i \in L_n$.
    \item[(ii)] The relevant $\alpha$-cuts $A^{y_{i_j}}$ are known for all $j = 1, \ldots, k$.
\end{itemize}

In other words, the membership function of $A$ is completely determined by its relevant $\alpha$-cuts, and vice versa.
\end{lema}

\begin{proof}
    The implication (i) $\Rightarrow$ (ii) is straightforward.

To prove implication (ii) $\Rightarrow$ (i), let $A^{y_{i_j}} = [l_j, r_j]$ with $l_j, r_j \in L_n$, for $j = 1, \ldots, k$, be the relevant $\alpha$-cuts of the dfn $A$.

The following algorithm reconstructs the membership values $A(i)$ for all $i \in L_n$:

\begin{algorithm}[H]
\caption{Computation of $A$ from its relevant $\alpha$-cuts}
\begin{algorithmic}[1]
\Require $n \in \mathbb{N}$, $m \ge 2$, and the relevant $\alpha$-cuts $A^{y_{i_j}} = [l_j, r_j]$
\For{$j = l_k$ to $r_k$}
  \State $A(j) \gets 1$ \Comment{Core computation}
\EndFor
\For{$i = k-1$ down to $1$}
    \For{$j = l_i$ to $l_{i+1} - 1$}
        \State $A(j) \gets y_{i_j}$
    \EndFor
    \For{$j = r_{i+1} + 1$ down to $r_i$}
        \State $A(j) \gets y_{i_j}$
    \EndFor
\EndFor
\end{algorithmic}
\end{algorithm}

Since $A = \bigcup_{j=1}^{k} A^{y_{i_j}}$, the above algorithm allows a discrete fuzzy number to be constructed from its relevant $\alpha$-cuts.
\end{proof}

\section{Motivation}
\label{sec:motivation}
In \cite{Gonzalez2025}, a method based on the \textit{pos function} (see Definition \ref{funcpos})  is introduced that allows implication functions to be obtained in the finite bounded lattice $(\subsetdfn, \orderDfnEq, \mathbf{1}_{\min}, \mathbf{1}_{\max})$  from discrete implication functions, as can be seen in the following result.

 \begin{theorem} Let $I$ be a discrete implication function on $L_k=\{0,\ldots, k\}$ where $k=\left|\subsetdfn\right|-1$. Then the following function:
    $$
\begin{array}{rccl}
{\mathcal I}:&\subsetdfn\times \subsetdfn &\longrightarrow&{\subsetdfn}\\
&(A,B)&\mapsto&{\mathcal I}(A,B)=pos^{-1}(I(pos(A),pos(B))),
\end{array}
$$
is an implication function on the lattice $\mathcal{L'}=(\subsetdfn, \orderDfnEq, \mathbf{1}_{\min}, \mathbf{1}_{\max})$.
\end{theorem}

Although the method is appropriate when the cardinality of the set $\left|\subsetdfn\right| $ is \textit{small}, the computation of the functions $\mathit{pos}$ and $\mathit{pos}^{-1}$ becomes computationally expensive as $m$ increases.

Indeed, if a brute-force approach is used in combination with the quicksort algorithm (see \cite[Chapter~8]{cormen2009introduction}), the time complexity is $O(N \log N) \approx O(m^{n+1} \log m)$.

Moreover, this approach requires storing all $N$ dfns in memory, which may be impractical for large values of $m$.

In this work, we propose a method that, on the one hand, avoids computing and storing all $N$ dfns, and on the other hand, computes the functions $\mathit{pos}$ and $\mathit{pos}^{-1}$ with a time complexity of order $O(n^2 m\log(n))$. In particular, for fixed $n$, we reduce the dependence on $m$ from $O(m^{n+1} \log m)$ (brute force with sorting) to $ O(n^2 m \log n)$.

\section{Computing the Functions $pos^{-1}$ and $pos$}
\label{sec:cost_inv_pos}

This section presents two algorithms that allow the $\mathit{pos}^{-1}$ and $\mathit{pos}$ functions to be calculated. In addition, its computational cost will be calculated, showing that this method allows us to efficiently obtain the discrete fuzzy number that occupies the i-th position in the finite bounded lattice $(\subsetdfn, \orderDfnEq, \mathbf{1}_{\min}, \mathbf{1}_{\max})$. 

\begin{proposition}\label{numintervals} Consider the finite chain $L_n=\{0,\ldots,n\}$ and let $\orderIntervals$ be a total order in the set $\setIntervals{n}$ of closed intervals defined on $L_n$. Then $(\setIntervals{n}, \orderIntervals)$ is a finite lattice where \[
\left| \setIntervals{n} \right| = \binom{n + 2}{2} = \frac{(n + 1)(n+2)}{2}.
\]
 In this way, if $I_i$ denotes the interval that occupies position $i$ in this finite lattice, we obtain the following chain of intervals:   \[
I_1 \orderIntervals I_2 \orderIntervals \ldots \orderIntervals I_{\frac{(n+1)(n+2)}{2}}
\] Note that this chain depends on the chosen total order.
\end{proposition}
\begin{proof}
   The reasoning is straightforward, since each interval $[a,b]\in  \setIntervals{n}$ is determined by two elements $a,b\in L_n$ satisfying $a\leq b$. 
\end{proof}

As will be shown, the algorithm to be proposed relies on the use of generalized discrete fuzzy numbers and some of their key properties. For this purpose, we first recall the concept of generalized discrete fuzzy number.
\begin{definition}\label{gdfn}
A fuzzy subset $A$ of $ \mathbb{R}$ with membership mapping $A:
\mathbb{R} \to [0,1]$ is called a \emph{generalized discrete fuzzy number}, or \emph{gdfn} for short, if
its support is finite, i.e., there exist $x_{1},...,x_{n}\in
\mathbb{R}$ with $x_{1}<x_{2}<...<x_{n}$ such that
$supp(A)=\{x_1,...,x_{n}\}$, and there are natural numbers $s, t$
with $1\leq s\leq t\leq n $ such that:
\begin{enumerate}
\item $A(x_{i})=k$ for all $i$ with $s\leq i\leq t$ with $k\leq 1$.
(\emph{k-core}) 
\item $A(x_{i})\leq A(x_{j})$ for all $i, j$ with  $1\leq i\leq j\leq s $. \item $A(x_{i})\geq
A(x_{j})$ for all $i, j$ with $t\leq i\leq j\leq n
$.
\end{enumerate}
A gdfn \(A\)  with support $supp(A)=\{x_1,\ldots,x_n\}$ is denoted as \( A = \{ A(1)/x_1, \ldots, A(n)/x_n \} \).

\end{definition}
In Proposition \ref{numintervals}, the total number of closed intervals that can be considered in the chain $L_n$ is calculated.  Based on this idea, we have the following definition:
\begin{definition}
Let us consider the set $\subsetdfn$, where $L_n = \{0, \ldots, n\}$ denotes the finite chain, and 
$Y_m = \{y_1 = 0 < y_2 < \ldots < y_m = 1\}$ represents the corresponding set of membership values.

For each $i \in \left\{1, \ldots, \frac{(n+1)(n+2)}{2} = \left| \setIntervals{n} \right| \right\}$, 
let $I_i = [a_i, b_i]$ denote the closed interval that occupies the $i$-th position in 
$(\setIntervals{n}, \orderIntervals)$.  
For each $j \in \{1, \ldots, m\}$, we define:
\begin{align*}
\setGdfn(a_i,b_i,j) 
&= \Biggl\{\, A \text{ is a gdfn with membership values in } 
Y_j = \{y_1 = 0 < \ldots < y_j\} \\
&\qquad\quad \Big|\, \{x \in L_n \mid A(x) = y_j\} = I_i = [a_i, b_i] \,\Biggr\}.
\end{align*}

In other words, $\setGdfn(a_i,b_i,j)$ denotes the set of gdfns whose membership values belong to 
$Y_j = \{y_1 = 0 < \ldots < y_j\}$, and whose maximum membership value corresponds to the interval 
$I_i = [a_i, b_i]$, which occupies position $i$ in the lattice $(\setIntervals{n}, \orderIntervals)$.
\end{definition}
\begin{proposition}
    The cardinality of the set $\setGdfn(i,j)$ is given by 
$$
|\setGdfn(a_i,b_i,j)| = \binom{a_i + j - 2}{j-2} \cdot \binom{n - b_i + j - 2}{j-2}.
$$
\end{proposition}
\begin{proof}
Determining the cardinality of $\setGdfn(a_i, b_i, j)$ is equivalent to computing, on the one hand, 
the number of increasing maps from the set $\{0, 1, \ldots, a_i - 1\}$ to 
$\{0 = y_1 < y_2 < \cdots < y_{j-1}\}$, and, on the other hand, 
the number of decreasing maps from $\{b_i + 1, \ldots, n\}$ to the same codomain.  

According to Lemma~3.2 in~\cite{MUNAR2022}, the number of monotonic maps from 
$\{1, \ldots, m_1\}$ to $\{1, \ldots, m_2\}$ is given, in general, by
\[
\binom{m_1 + m_2 - 1}{m_2 - 1}.
\]
Hence, the number of increasing maps from $\{0, 1, \ldots, a_i - 1\}$ to 
$\{y_1, \ldots, y_{j-1}\}$ is 
\[
\binom{a_i + j - 2}{j - 2},
\]
and the number of decreasing maps from $\{b_i + 1, \ldots, n\}$ to the same codomain is 
\[
\binom{n - b_i + j - 2}{j - 2}.
\]
Since these choices are independent, the total number of gdfns in $\setGdfn(a_i, b_i, j)$ is the product of both quantities, that is,
\[
\left|\setGdfn(a_i, b_i, j)\right| = 
\binom{a_i + j - 2}{j - 2} 
\binom{n - b_i + j - 2}{j - 2}.
\]
This expression provides the desired result.
\end{proof}

From the previous discussion, we obtain the following algorithm for computing the dfn $A$ in the lattice $(\subsetdfn, \orderDfnEq)$ given an index $i$ with  $0\leq i \leq \left|\subsetdfn\right|-1$:

\begin{algorithm}[H]
\caption{Computation of dfn $A$ given an index $i$}
\begin{algorithmic}[1]
\Require $n \in \mathbb{N}$, $m \ge 2$, and an index $i$
\State $A^{y_{m+1}} \gets \emptyset$
\State $k_{m+1} \gets \frac{(n+1)(n+2)}{2}$
\For{$j = m$ down to $1$}
  \State $C_j \gets \{I_l \in \setIntervals{n} \mid A^{y_{j+1}} \subseteq I_l,\ l \leq k_{j+1}\} = \{I_1^{(j)}\orderIntervalsEstrict\ldots\orderIntervalsEstrict I_{s_j}^{(j)}\}$
  \State $k_j \gets \min\{l \mid I_l^{(j)} = [a_l^{(j)},b_l^{(j)}] \in C_j,\ \sum\limits_{k=1}^l |\setGdfn(a_k^{(j)},b_k^{(j)},j)| \geq i\}$
  \State $A^{y_j} \gets I_{k_j}$
  \State $i \gets i - \sum\limits_{k=1}^{k_j-1}|\setGdfn(a_{k}^{(j)}, b_{k}^{(j)}, j)|$, where $I_{k}^{(j)} = [a_{k}^{(j)}, b_{k}^{(j)}]$
\EndFor
\end{algorithmic}
\label{unrank-by-cuts}
\end{algorithm}
Before presenting the complexity analysis, we provide a brief explanation of the algorithm used to compute the dfn $A$ within the lattice $(\subsetdfn, \orderDfnEq)$. 
The algorithm proceeds as follows.

We start with the highest membership value of the dfn $A$, namely $y_m = 1$. 
At this level, we compute the set $C_m$, which contains all intervals from $\setIntervals{n}$ whose index is less than or equal to 
$k_{m+1} = \frac{(n+1)(n+2)}{2}$, i.e., the total number of intervals in $\setIntervals{n}$.

Next, we determine the index of the smallest interval, according to the considered interval order, such that the sum of the cardinalities of the sets of gdfns corresponding to all smaller intervals exceeds the index $i$. 
This interval becomes the core of the target dfn.

We then proceed downward through the levels, considering $y_{m-1}, y_{m-2}, \ldots, y_1$. 
At each level $j$, we repeat the same process: compute the corresponding interval, update the value of $i$, and determine the $\alpha$-cut $A^{y_j}$ associated with that level.

According to Lemma~\ref{EQALPHACUTS}, computing all the $\alpha$-cuts of a dfn $A$ is equivalent to reconstructing the entire dfn.
\begin{proposition}\label{costcomput}
The computational cost of the Algorithm 2 is  $O(n^2 m \log(n))$.  
\end{proposition}
\begin{proof}
    
Since $|\setIntervals{n}| = \frac{(n+1)(n+2)}{2} = O(n^2)$, the computation of each set $C_j$ and the corresponding index $k_j$ requires sorting or searching among $O(n^2)$ elements.  
If the quicksort algorithm is employed, this operation has a cost of 
$O(n^2 \log(n^2)) = O(n^2 \log n)$ per iteration.

Repeating this process at most $m$ times—once for each membership level $y_j$—the total computational cost of the algorithm is therefore at most:
\[
O(n^2 m \log n).
\]

\end{proof}

\label{sec:cost_pos}

A similar argumentation leads to the following algorithm for the $pos$ function. 
In this case, given a discrete fuzzy number $A \in \subsetdfn$, 
we aim to compute the index $i \in \{0, \ldots, N-1\}$ such that $A$ occupies position $i$ in the finite lattice $(\subsetdfn, \orderDfnEq)$.

To do this, we assume that the $\alpha$-cuts $A^{y_j}$ for $j = 1, \ldots, m$ are known.

The algorithm that computes the index $i$ in the lattice $(\subsetdfn, \orderDfnEq)$ is as follows:

\begin{algorithm}[H]
\caption{Computation of index $i$ given a dfn $A$}
\label{alg:cuts-rank}
\begin{algorithmic}[1]
\Require $n \in \mathbb{N}$, $m \ge 2$, and a dfn $A$
\State $i \gets 1$
\State $A^{y_{m+1}}=\emptyset$
\For{$j = m$ down to $1$}
    \State $C_j \gets \{I_l \in \setIntervals{n} \mid A^{y_{j+1}} \subseteq I_l,\ I_l\orderIntervals A^{y_j}\} = \{I_1^{(j)}\orderIntervalsEstrict\ldots\orderIntervalsEstrict I_{s_j}^{(j)}\}$
  \State $i \gets i + \sum\limits_{k=1}^{s_j-1} |\setGdfn(a_k^{(j)},b_k^{(j)},j)|$
\EndFor
\end{algorithmic}
\label{rank-i-dfn}
\end{algorithm}
Finally, the last result of this section gives the computational cost of previous algorithm. 
\begin{proposition}
The computational cost of the Algorithm 3 is  $O(n^2 m \log(n))$.  
\end{proposition}
\begin{proof}
Following the same argument as in the Proposition \ref{costcomput}, it is straightforward to show that the computational cost of this algorithm is $O(n^2 m \log(n))$.
\end{proof}
The following examples provide a detailed illustration of how the algorithms presented in this section operate.

\begin{example}\label{ex:academic_example}

Let us consider the lattice $(\mathcal{D}_1^{L_5\rightarrow Y_6}, \Delta^{\downarrow}_{\delta})$, where  $L_5=\{0,1,2,3,4,5\}$, $Y_6=\{y_1=0<y_2=0.2<y_3=0.4<y_4=0.6<y_5=0.8<y_6=1\}$ and $\Delta^{\downarrow}_{\delta}=t\text{-inc}$ is the intervalar order (see Definition \ref{deforder} and corresponding examples). 

According to Proposition \ref{numintervals}, we have $\left| \setIntervals{5} \right| = \binom{5 + 2}{2} = \frac{(5 + 1)(5+2)}{2}= 21$ and using the $t\text{-inc}$ order the closed intervals of this lattice $(\left| \setIntervals{5} \right|,t\text{-inc})$ are sorted as follows:
\begin{align*}
[0,5] &\orderIntervalsEstrict [0,4] \orderIntervalsEstrict [0,3] \orderIntervalsEstrict [0,2] \orderIntervalsEstrict [0,1] \orderIntervalsEstrict [0,0]\\
&\orderIntervalsEstrict [1,5] \orderIntervalsEstrict [1,4] \orderIntervalsEstrict [1,3] \orderIntervalsEstrict [1,2] \orderIntervalsEstrict [1,1]\\
&\orderIntervalsEstrict [2,5] \orderIntervalsEstrict [2,4] \orderIntervalsEstrict [2,3] \orderIntervalsEstrict [2,2]\\
&\orderIntervalsEstrict [3,5] \orderIntervalsEstrict [3,4] \orderIntervalsEstrict [3,3] \orderIntervalsEstrict [4,5] \orderIntervalsEstrict [4,4] \orderIntervalsEstrict [5,5].
\end{align*}

Moreover, taking into account Proposition \ref{numerodedfn} we have the cardinality of \[
\left| \mathcal{D}_1^{L_5\rightarrow Y_6}\right| = \binom{5 + 12 - 2}{12 - 2} = \frac{(5 + 12 - 2)!}{(12 - 2)! \, 5!}=\binom{15}{10}=3003.\]

To show how Algorithm 2 works, let us suppose that we want to know the discrete fuzzy number of lattice $ (\mathcal{D}_1^{L_5\rightarrow Y_6}, t\text{-inc})$ that occupies position 50, that is, we want to know the value of $pos^{-1}(50)$. 

We follow Algorithm~\ref{unrank-by-cuts} step by step:

\paragraph{Initialization and first iteration}

\begin{enumerate}[label=\arabic*.]
    \item Initialize: $I_7 = \emptyset$
    \item Set $k_7 = 21$
    \item Let $j = 6$
    \item The set of candidate intervals is:
    \[
    C_6 = \{[0,5] \prec [0,4] \prec [0,3] \prec [0,2] \prec [0,1] \prec [0,0] \prec [1,1] \prec \ldots \prec [5,5]\}
    \]
    \item For each interval $[a_j, b_j]$, we compute:
    \[
    |\setGdfn(a_j, b_j, 6)| = \binom{a_j + 4}{4} \binom{5 - b_j + 4}{4} = \binom{a_j + 4}{4} \binom{9 - b_j}{4}
    \]

    The relevant values and their accumulated sums are:

    \begin{center}
    \begin{tabular}{|c|c|c|c|}
    \hline
    Index & Interval & $|\setGdfn|$ & Accumulated \\
    \hline
    1 & $[0,5]$ & 1 & 1 \\
    2 & $[0,4]$ & 5 & 6 \\
    3 & $[0,3]$ & 15 & 21 \\
    4 & $[0,2]$ & 35 & 56 \\
    \hline
    \end{tabular}
    \end{center}

    \item We find $k_6 = 4$, since it's the smallest $l$ such that:
    \[
    \sum_{k=1}^{l} |\setGdfn(a_k^{(6)}, b_k^{(6)}, 6)| \geq i = 50
    \]

    \item Set the $\alpha$-cut:
    \[
    A^{y_6} = A^1 = \text{core}(A) = I_4 = [0,2]
    \]

    \item Update the index:
    \[
    i \gets 50 - \sum_{k=1}^{3} |\setGdfn(a_k^{(6)}, b_k^{(6)}, 6)| = 50 - 21 = 29
    \]
\end{enumerate}

\paragraph{Next Iteration}

\begin{enumerate}[label=\arabic*.]
\setcounter{enumi}{8}
    \item Let $j = 5$
    \item Candidate intervals:
    \[
    C_5 = \{[0,5] \prec [0,4] \prec [0,3] \prec [0,2]\}
    \]
    \item Compute:
    \[
    |\setGdfn(a_j, b_j, 5)| = \binom{a_j + 3}{3} \binom{5 - b_j + 3}{3} = \binom{a_j + 3}{3} \binom{8 - b_j}{3}
    \]

    Table of values:

    \begin{center}
    \begin{tabular}{|c|c|c|c|}
    \hline
    Index & Interval & $|\setGdfn|$ & Accumulated \\
    \hline
    1 & $[0,5]$ & 1 & 1 \\
    2 & $[0,4]$ & 4 & 5 \\
    3 & $[0,3]$ & 10 & 15 \\
    4 & $[0,2]$ & 20 & 35 \\
    \hline
    \end{tabular}
    \end{center}

    \item We find $k_5 = 4$, since:
    \[
    \sum_{k=1}^{4} |\setGdfn(a_k^{(5)}, b_k^{(5)}, 5)| \geq i = 29
    \]

    \item Set:
    \[
    A^{y_5} = A^{0.8} = [0,2]
    \]

    \item Update:
    \[
    i \gets 29 - \sum_{k=1}^{3} |\setGdfn(a_k^{(5)}, b_k^{(5)}, 5)| = 29 - 15 = 14
    \]
\end{enumerate}

\paragraph{Remaining $\alpha$-cuts}

\[
A^{y_4} = A^{0.6} = [0,2], \quad
A^{y_3} = A^{0.4} = [0,3], \quad
A^{y_2} = A^{0.2} = [0,5], \quad
A^{y_1} = A^{0} = [0,5].
\]

\paragraph{Final dfn}
Taking into account the previous result and Lemma \ref{descomposicio_dfn}, we obtain 
\[
A = \{1/0, 1/1, 1/2, 0.4/3, 0.2/4, 0.2/5\}\in \mathcal{D}_1^{L_5\rightarrow Y_6}
\]
\end{example}
The following example illustrates how Algorithm 3 works.
\begin{example}
\label{ex:academic_example2}
In this case, we will consider the same lattice and total order as established in the previous example. Thus, let us consider the lattice $(\mathcal{D}_1^{L_5\rightarrow Y_6}, \Delta^{\downarrow}_{\delta})$, where  $L_5=\{0,1,2,3,4,5\}$, $Y_6=\{y=1=0<y_2=0.2<y_3=0.4<y_4=0.6<y_5=0.8<y_6=1\}$ and $\Delta^{\downarrow}_{\delta}=t\text{-inc}$ is the intervalar order. 

In this case, given a discrete fuzzy number~$A\in \mathcal{D}_1^{L_5\rightarrow Y_6} $, we want to compute the index $i$ such that $\mathrm{pos}^{-1}(i) = A$, or equivalently, $\mathrm{pos}(A) = i$.

Let us consider the following discrete fuzzy number:
\[
A = \{1/0,\, 1/1,\, 1/2,\, 0.2/3,\, 0/4,\, 0/5\}\in \mathcal{D}_1^{L_5\rightarrow Y_6}
\]

The relevant $\alpha$-cuts are:
\[
\begin{aligned}
A^{y_6} &= A^{1} = [0,2], \quad
A^{y_5} = A^{0.8} = [0,2], \quad
A^{y_4}  = A^{0.6} = [0,2], \\
A^{y_3} & = A^{0.4} = [0,2], \quad
A^{y_2} = A^{0.2} = [0,3], \quad
A^{y_1} = A^{0} = [0,5].
\end{aligned}
\]

We now follow Algorithm~\ref{rank-i-dfn} step by step.

\paragraph{Initialization and first iteration}

\begin{enumerate}
    \item Initialize $i \gets 1$
    \item Initialize $A^{y_7} \gets \emptyset$
    \item Let $j = 6$
    \item Candidate intervals:
    \[
    C_6 = \left\{ I_l \in \mathcal{I}_5 \,\middle|\, I_l \prec A^{y_6} = [0,2] \right\} = \{[0,5] \prec [0,4] \prec [0,3] \prec [0,2]\}
    \]
    \item For each interval $[a_j, b_j]$, compute:
    \[
    |\setGdfn(a_j, b_j, 6)| = \binom{a_j + 4}{4} \binom{9 - b_j}{4}
    \]
    From previous computations, $s_6 = 4$
    \item Update:
    \[
    i \gets i + \sum_{k=1}^3 |\setGdfn(a_k^{(6)}, b_k^{(6)}, 6)| = 1 + 21 = 22
    \]
\end{enumerate}

\paragraph{Next Iteration ($j = 5$)}

\begin{enumerate}[label=\arabic*.]
\setcounter{enumi}{6}
    \item Candidate intervals:
    \[
    C_5 = \left\{ I_l \in \mathcal{I}_5 \,\middle|\, A^{y_6} \subseteq I_l,\ I_l \prec A^{y_5} = [0,2] \right\} = \{[0,5] \prec [0,4] \prec [0,3] \prec [0,2]\}
    \]
    \item Compute:
    \[
    |\setGdfn(a_j, b_j, 5)| = \binom{a_j + 3}{3} \binom{8 - b_j}{3}
    \]
    From previous computations, $s_5 = 4$
    \item Update:
    \[
    i \gets 22 + \sum_{k=1}^3 |\setGdfn(a_k^{(5)}, b_k^{(5)}, 5)| = 22 + 15 = 37
    \]
\end{enumerate}

The following iterations are:

\begin{center}
\begin{tabular}{|c|c|}
\hline
\textbf{Value of $j$} & \textbf{Value of $i$} \\
\hline
4 & 47 \\
3 & 53 \\
2 & 55 \\
1 & 55 \\
\hline
\end{tabular}
\end{center}

Therefore, the final value of $i$ is:
$i = 55$, that is, the position of discrete fuzzy number \[
A = \{1/0,\, 1/1,\, 1/2,\, 0.2/3,\, 0/4,\, 0/5\}
\] in the lattice $(\mathcal{D}_1^{L_5\rightarrow Y_6}, \Delta^{\downarrow}_{\delta})$ is the 55th.
\end{example}

\section{Empirical scaling with fixed $n$}
\label{sec:empirical}
To validate the theoretical time complexity derived in Section~\ref{sec:cost_inv_pos}, we conducted empirical experiments assessing the scalability of the proposed ranking and unranking algorithms under varying parameter sizes.

\subsection{Experimental Setup and Numerical Results}

This subsection provides the numerical
results obtained from the empirical evaluation of the proposed ranking and unranking algorithms.
All experiments were implemented in Python~3.11 and executed on an Intel i7--12700H CPU 
(2.7\,GHz, 16\,GB RAM). 

We have selected the interval-based order t-inc, as defined in Subsection~\ref{sub:partial_total_orders}.
For each configuration, we fixed $n = 10$ and varied $m$ in the range ${100, 200, \ldots, 1000}$, that is, from $m = 100$ to $m = 1000$ in increments of $100$. For every value of $m$, we performed $K = 500$ independent 
trials using uniformly sampled indices $i \in \{0, \ldots, N-1\}$, where $N = |\subsetdfn|$ 
denotes the total number of discrete fuzzy numbers of the set $\subsetdfn$. 

Each reported value corresponds to the sample mean $\bar{t}_m$ of the $K$ trials, accompanied 
by its standard deviation $s_m$. The standard error of the mean is given by $s_m / \sqrt{K}$, 
so the relative uncertainty decays as $1 / \sqrt{K}$. Table~\ref{tab:timing_results} summarizes the
average execution times (in milliseconds) for different values of $m$, providing a quantitative
view of the algorithm’s scalability. 

As shown in  Table \ref{tab:timing_results}, the average runtime grows almost linearly with $m$, in agreement with
the theoretical time complexity $O(m)$ for fixed $n$ derived in Section ~\ref{sec:cost_inv_pos}.

\begin{table}[h!]
\centering
\caption{Average execution time as a function of $m$ (for fixed $n=10$). Results over $K=500$ trials; order=t-inc.}
\label{tab:timing_results}
\begin{tabular}{r r r r r}
\toprule
   $m$ &  unrank $\bar{t}_m$ (ms) &  unrank sm (ms) &  rank $\bar{t}_m$ (ms) &  rank sm (ms) \\
\midrule
 100 &            0.870 &           0.172 &          0.810 &         0.173 \\
 200 &            1.713 &           0.322 &          1.565 &         0.310 \\
 300 &            2.660 &           0.503 &          2.426 &         0.481 \\
 400 &            3.554 &           0.768 &          3.255 &         0.742 \\
 500 &            4.486 &           0.932 &          4.113 &         0.888 \\
 600 &            5.392 &           1.080 &          4.955 &         1.049 \\
 700 &            6.253 &           1.300 &          5.746 &         1.263 \\
 800 &            7.207 &           1.396 &          6.606 &         1.349 \\
 900 &            7.924 &           1.600 &          7.272 &         1.553 \\
1000 &            8.960 &           1.822 &          8.205 &         1.739 \\
\bottomrule
\end{tabular}
\end{table}

These results empirically confirm the theoretical $O(n^2\log(n) m)$ bound, and show that for a fixed 
$n$ the dependence on $m$ is effectively linear. Both ranking and unranking procedures therefore 
exhibit predictable and scalable performance, validating the efficiency of the proposed method.

\subsection{Graphical Analysis and Empirical Scaling Behavior}

Figure~\ref{fig:scaling_plots} illustrates the empirical scaling behavior of the proposed
ranking and unranking algorithms for fixed $n = 10$ and varying $m$.
Panel~(a) shows the average execution time $\bar{t}_m$ (in milliseconds) as a function of $m$,
while panel~(b) presents the same data in a log--log plot together with a linear regression fit for both ranking and unranking strategies.

The approximately unit slope observed in panel~(b) confirms the expected linear dependence
of the runtime on $m$ for fixed $n$, as predicted by the theoretical analysis in 
Section~\ref{sec:cost_inv_pos}. Minor deviations from perfect linearity for small values of $m$
are attributed to fixed computational overheads (e.g., interpreter latency, memory allocation, 
and cache effects) that dominate when the problem size is small.
As $m$ increases, these constants become negligible and the empirical slope converges 
to the theoretical value of~1.

\begin{figure}[t]
\hspace{-1cm}
  \begin{minipage}{0.49\linewidth}
   \includegraphics[width=1.1\linewidth]{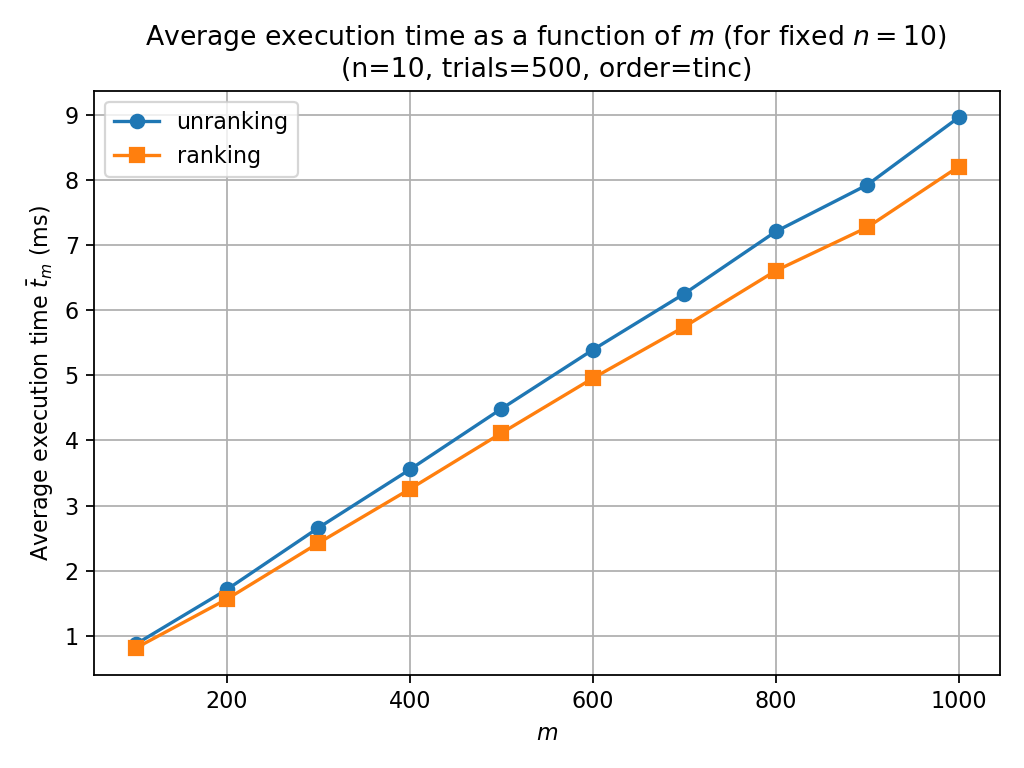}
    \subcaption{}
    \label{fig:times-vs-m}
  \end{minipage}
  \hfill
  \begin{minipage}{0.49\linewidth}
    \includegraphics[width=1.1\linewidth]{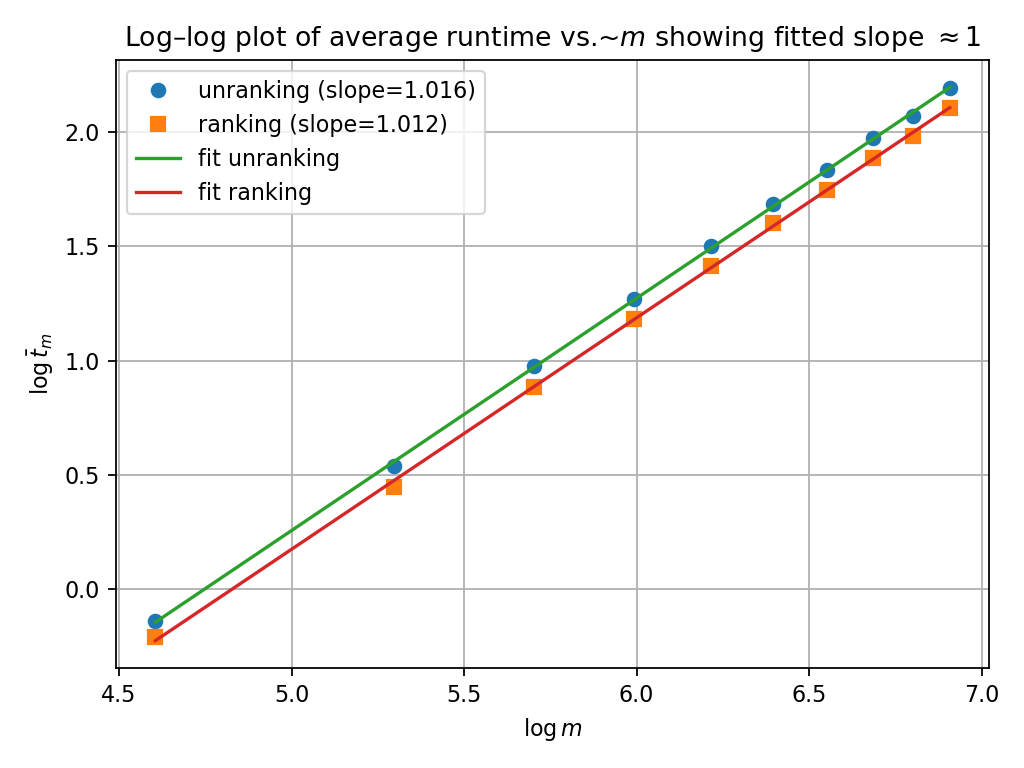}
    \subcaption{}
    \label{fig:loglog}
  \end{minipage}
  \caption{(a) Empirical scaling of ranking and unranking algorithms for fixed $n=10$, and (b) Log--log plot with linear fit showing slope $\approx 1$.}
  \label{fig:scaling_plots}
\end{figure}

Overall, these results corroborate the theoretical $O(n^2\log(n) m)$ bound, and demonstrate that,
for a fixed $n$, the proposed algorithms scale linearly with $m$ both in practice and in theory.

\section{Conclusions and Future Work}\label{sec:conclusions}
Motivated by recent methods \cite{Gonzalez2025,DeMiguel2024} for constructing implication and aggregation functions in the set \( \subsetdfn \) based on analogous operators defined on a finite chain \( L_k \), where \( k = \left| \subsetdfn \right| - 1 \), and by the fact that once a total (admissible) order \( \orderDfnEq \) is fixed, a bijective application
\[
\begin{array}{rccl}
pos: & \subsetdfn & \longrightarrow & L_k \\
     & A          & \mapsto        & pos(A)=\left|\{X \in \subsetdfn \mid X \orderDfnEq A\}\right| - 1,
\end{array}
\]
can be established, which enables us to determine the position of a discrete fuzzy number \( A \in (\subsetdfn, \orderDfnEq) \), as well as its inverse, which allows us to obtain, for any specific position \( 0\leq j \leq k \), the discrete fuzzy number \( A \in \subsetdfn \) that corresponds to that position. For this reason, in this paper, we have presented a deterministic, order-consistent algorithms for modeling the \textit{pos} bijection. This construction turns the abstract lattice structure into an operational indexing tool, allowing direct access, enumeration, and reconstruction of discrete fuzzy numbers without loss of generality. 

Two mutually inverse procedures have been developed: the\textit{ ranking function} $\mathrm{pos}$, which assigns to each dfn its exact position in the ordered lattice, and \textit{the unranking function} $\mathrm{pos}^{-1}$, which reconstructs the dfn from a given index. Both rely on closed-form combinatorial counts over nested $\alpha$-cuts, avoiding exhaustive enumeration while preserving exactness. In computational terms, the algorithms run in deterministic time $O(n^2\log n \, m)$ with $O(1)$ extra memory; empirical results confirm the near-linear scaling predicted in $m$ for fixed $n$, thus validating the practical efficiency of the approach. 

As future work, we plan to extend the proposed framework to the domain of discrete 
Z-numbers. Since a Z-number consists of a pair of fuzzy descriptors---one representing 
the value component and the other the reliability component---ordering such objects 
requires simultaneously comparing two structured sources of information. Our intention 
is to incorporate the total orders studied in~\cite{Mir-Fuentes2024} into this 
setting, adapting them so that they remain consistent with the two-level semantics of 
Z-numbers.

The idea is to construct an efficient ranking--unranking mechanism on the space of 
discrete Z-numbers, in analogy with the $\mathrm{pos}$--$\mathrm{pos}^{-1}$ methodology 
developed in this work for classical discrete fuzzy numbers. To achieve this, we may 
need to restrict the membership values of one or both components to finite sets 
(similarly to the set $Y_m$ used throughout the paper), ensuring that the resulting 
combinatorial structure remains tractable. Under these restrictions, it should be 
possible to characterize the cardinality of the admissible family of Z-numbers and to 
derive closed-form counting formulas analogous to those obtained for discrete fuzzy 
numbers.

Such a ranking--unranking scheme would enable the construction of implication and 
aggregation functions directly on the ordered lattice of discrete Z-numbers, extending 
the approaches in~\cite{Gonzalez2025,DeMiguel2024}. We expect this extension to 
provide a unified and computationally efficient framework for modeling Z-valued 
information in decision-making and linguistic reasoning contexts.

\backmatter

\bmhead{Supplementary information}
To support full reproducibility of our results, we provide an open-source GitHub repository that contains the implementation of the ranking and unranking algorithms, as well as all experimental scripts and datasets: \url{https://github.com/AlejandroMus/dfn-ranking-unranking}.

\bibliography{sn-bibliography}

\end{document}